\newif\ifshortver
\shortvertrue
\shortverfalse

\ifshortver
\documentclass[10pt, conference, letterpaper]{IEEEtran}
\else
\documentclass[english, onecolumn]{IEEEtran}
\fi

\IEEEoverridecommandlockouts
\usepackage{algorithm}
\usepackage[noend]{algpseudocode}

\algrenewcommand\algorithmicindent{5mm}

\usepackage{amsmath}
\allowdisplaybreaks
\usepackage{amssymb,amsthm}
\usepackage{amsfonts}
\usepackage{cite}
\usepackage{tikz}
\usepackage{bbm}
\usetikzlibrary{positioning,arrows,shapes,chains,fit,scopes}
\usepackage{pgfplots}
\usepackage{pgfplotstable}	
\usepackage{booktabs}
\usepackage{colortbl}
\usepackage{color} 
\usepackage{blkarray, bigstrut}
\usepackage{amsmath}
\DeclareMathOperator*{\argmax}{arg\,max}

\theoremstyle{theorem}
\newtheorem{theorem}{Theorem}
\newtheorem{lemma}[theorem]{Lemma}

\theoremstyle{definition}
\newtheorem{definition}{Definition}
\newtheorem{example}{Example}
\theoremstyle{remark}
\newtheorem{remark}{Remark}

\newcommand{\conv}{\mathrm{conv}}

\newcommand{\mc}[1]{\mathcal{#1}}

\newcommand{\dR}{\mathcal{R}}

\tikzstyle{src}=[circle,draw=gray!80!black,fill=gray!10,thick,inner
sep=1pt,minimum size=12pt]
\tikzstyle{dst}=[circle,draw=gray!80!black,fill=gray!10,thick,inner
sep=1pt,minimum size=12pt]
\tikzstyle{mid}=[circle,draw=gray!80,fill=gray!10,thick,inner
sep=1pt,minimum size=12pt]

\title{Joint Scheduling and Multiflow Maximization in Wireless Networks}
\begin{document}
\ifshortver
\author{
  \IEEEauthorblockN{Yanxiao Liu$^{\dagger}$, Shenghao Yang$^{*}$ and Cheuk Ting Li$^{\dagger}$}
  \IEEEauthorblockA{%
    $^{\dagger}$Department of Information Engineering, The Chinese University of Hong Kong\\
    $^{*}$School of Science and Engineering, The Chinese University of Hong Kong, Shenzhen\\
    Email: yanxiaoliu@link.cuhk.edu.hk, shyang@cuhk.edu.cn, ctli@ie.cuhk.edu.hk}
}
\else
\author{
Yanxiao Liu, \textit{Graduate Student Member, IEEE}, Shenghao Yang,  \textit{Member, IEEE} and Cheuk Ting Li, \textit{Member, IEEE}
\thanks{
Y. Liu is with the Department of Information Engineering, The Chinese University of Hong Kong, Hong Kong, China. Email: yanxiaoliu@link.cuhk.edu.hk

S. Yang is with the School of Science and Engineering, The Chinese
University of Hong Kong, Shenzhen, Shenzhen 518172, China. Email:
shyang@cuhk.edu.cn

C. T. Li is with the Department of Information Engineering, The Chinese
University of Hong Kong, Hong Kong, China. Email: ctli@ie.cuhk.edu.hk.
}
}
\fi

\maketitle 

\begin{abstract}

Towards the development of 6G mobile networks, it is promising to integrate a large number of devices from multi-dimensional platforms, and it is crucial to have a solid understanding of the theoretical limits of large-scale networks. 
We revisit a fundamental problem at the heart of network communication theory: the maximum multiflow (MMF) problem in multi-hop networks, with network coding performed at intermediate nodes. 
To derive the exact-optimal solution to the MMF problem (as opposed to approximations), conventional methods usually involve two steps: first calculate the scheduling rate region, and then find the maximum multiflow that can be supported by the achievable link rates.
However, the NP-hardness of the scheduling part makes solving the MMF problem in large networks computationally prohibitive. 
In this paper, while still focusing on the exact-optimal solution, we provide efficient algorithms that can jointly calculate the scheduling rate region and solve the MMF problem, thereby outputting optimal values \emph{without} requiring the entire scheduling rate region. 
We theoretically prove that our algorithms always output optimal solutions in a finite number of iterations, and we use various simulation results to demonstrate our advantages over conventional approaches.
Our framework is applicable to the most general scenario in multi-source multi-sink networks: the \emph{multiple multicast} problem with network coding. 
Moreover, by employing a graphical framework, we show that our algorithm can be extended to scenarios where propagation delays are large (e.g., underwater networks), in which recent studies have shown that the scheduling rate region can be significantly improved by utilizing such delays. 
\end{abstract}

\ifshortver
\else
\fi

\section{Introduction}
\label{sec::intro}

Over the past years, both the number of users in wireless networks and the services to be provided have experienced significant growth. 
To deploy the next generation mobile system, it is expected that a massive connectivity and emerging applications should be supported. 
For the design of large-scale, highly-connected wireless mobile systems, it is crucial to have a solid understanding of their theoretic limits, and hence we revisit two multiflow optimization problems: the maximum multiflow (MMF) and maximum concurrent multiflow (MCMF) that can be supported by collision-free link schedules of the networks.

These multiflow problems maximize the total or concurrent throughput between multiple source nodes and sink nodes supported by achievable link rates. 
To solve these problems, traditional methods typically first model wireless interference in networks using a link conflict graph~\cite{jain2005impact}. Next, they compute the scheduling rate region, which is a well-known NP-hard problem that is also difficult to approximate~\cite{baker1994approximation}. Finally, they determine the maximum total or concurrent throughput under the constraints defined by the scheduling rate region. 
Due to the hardness of the scheduling problem, it is unrealistic to directly solve the multiflow maximization problems in large-scale networks in this way. 
In~\cite{wan2009multiflows} it has been shown that both and MMF and MCMF problems are NP-hard even in very simple settings. 
They can become even harder when network coding~\cite{ahlswede2000network, li2003linear, yeung2008information} is performed. 
Therefore, joint optimization methods usually focus on either approximate solutions of general networks or exact solutions of restricted networks (see Section~\ref{subsec::review_opt} for a review).

In this paper, from a theoretic perspective, we are interested in the \emph{exact-optimal} solutions of the multiflow problems for \emph{general} multi-hop networks. 
Rather than first solving the scheduling problem and then calculating the multiflow values (referred to as \emph{two-step algorithms}), we design efficient algorithms that jointly calculate the maximum total or concurrent multiflow and the scheduling rate region by employing a decomposition method, thereby requiring only a (possibly very small) subset of the scheduling rate region. 
This joint framework makes our algorithms more practical, while still \emph{provably} guaranteeing optimal (not approximate or converging-to-optimal) solutions to the multiflow maximization problems in a finite number of iterations. 
The efficiency of our approach is demonstrated through simulation results on various network settings. 
Our algorithms are applicable to the most general setting in multi-hop networks: the \emph{multiple multicast} case, where network coding~\cite{ahlswede2000network, li2003linear, yeung2008information} at intermediate nodes is allowed.

Moreover, recent studies~\cite{ma2021rate, yang2023wireless_tit, bai2017throughput, chitre2012throughput, hsu2009st} have shown that in challenging environments (e.g., underwater and deep-space networks), where the propagation delays are non-negligible, the scheduling rate region can be significantly enlarged by utilizing those delays, though it also becomes even harder to compute. 
Benefiting from the generality of our algorithm, we also extend our joint framework to such cases by employing a graphical model, while preserving all its advantages: generality, efficiency, and optimality.

\begin{remark}
    Our algorithms rely on an integer linear programming step and, therefore, do not theoretically guarantee polynomial-time complexity, which is expected due to the NP-hardness of the MMF or MCMF problem~\cite{jain2005impact, wan2009multiflows}. 
    This can be understood in that, in the worst case, one may still have to compute the entire scheduling rate region. 
    However, we use simulation results to demonstrate the efficiency of our algorithms, where the gain comes from the fact that usually only a small subset of the scheduling rate region is needed. 
    Our other contributions mainly lie in the generality and optimality of our framework.
    The closest work to ours is~\cite{traskov2012scheduling}; see Remark~\ref{remark2} and Section~\ref{subsec::review_opt} for the differences.
\end{remark}

The remainder of the paper is organized as follows. 
We first provide a comprehensive literature review in Section~\ref{sec::review}. 
In Section~\ref{sec::no_delay_net}, we describe the network model and problem formulation. 
We propose our algorithm in Section~\ref{sec::alg} and prove the optimality. 
We extend the framework to networks with non-negligible propagation delays in Section~\ref{sec::delay_net}. 
The performance evaluation is in Section~\ref{sec::simu}.

\section{Related Works}
\label{sec::review}

The literature review includes three parts: 
First, we discuss the multiflow problems. 
Second, we review related optimization frameworks and network coding. 
Finally, we review recent works on utilizing non-negligible propagation delays in certain (e.g., underwater and deep-space) networks.

\subsection{Maximum Multiflow Problem}

The maximum multiflow (MMF) and the maximum concurrent multiflow (MCMF) are core problems in the theory of network communications. 
The MMF problem studies the maximum throughput between selected source nodes and sink nodes~\cite{wan2009multiflows, jain2005impact}, and the maximum concurrent multiflow  problem~\cite{shahrokhi1990maximum} models the case where every sender-receivers session transmits messages concurrently. 
The NP-hardness of both problems have been proved in~\cite{wan2009multiflows}, even in very simple network settings. 
In~\cite{kodialam2003characterizing, kodialam2005characterizing}, both the MMF and MCMF problems are discussed under the interference model that nodes cannot transmit and receive simultaneously. 
By enforcing interference constraints on links, \cite{kumar2005algorithmic} guarantees the schedulability and develops constant-approximation algorithms. 
More linear programming formulation and approximation algorithms can be found in~\cite{shahrokhi1990maximum, kodialam2003characterizing, wan2009multiflows}. 
In~\cite{su2008characterizing}, the MMF and MCMF problems are discussed by dividing the cases to full-duplex systems and half-duplex systems, both of which are covered by our interference model in this paper. 
The MMF problem has been extended to unicast networks with network coding~\cite{zhou2013maximum}, where the network coding~\cite{ahlswede2000network, li2003linear, yeung2008information} is treated as a scheme to decrease the impact of wireless interference.

To support the (concurrent) multiflow in networks, it is usually required to first find achievable link rates, which forms the scheduling rate region. 
To solve the scheduling problem, in~\cite{jain2005impact}, the effects of wireless interference can be modeled by a \emph{conflict graph}, and the scheduling problem is equivalent to searching all the maximal independent sets in the conflict graph, which is an NP-hard problem that is also hard to approximate~\cite{baker1994approximation}. 
Due to this hardness, it is computationally prohibitive to calculate the scheduling region before solving the MMF (or MCMF) problem. 
Though our focus is on algorithms that guarantee optimal solutions for general networks, other (decentralized) optimization algorithms with low complexity can be used in practice for either approximate results on general networks or optimal results on restricted networks, which are reviewed as follows.

\subsection{Joint Optimization and Network Coding} 
\label{subsec::review_opt}

Existing works on the MMF (or MCMF) problem~\cite{jain2005impact, wan2009multiflows, zhou2013maximum} only study the \emph{multiple unicast} case, i.e., each source node is paired with one sink node. 
However, we consider \emph{multiple multicast} in general multi-source multi-sink networks, where each of a number of source nodes transmits a message to a set of sink nodes. 
In this scenario, network coding~\cite{ahlswede2000network, li2003linear, yeung2008information} is an effective technique to improve the network performance, and the throughput can increase up to several folds~\cite{katti2006xors, yan2012implicit}.  
The joint consideration of throughput, scheduling and network coding has been widely studied in~\cite{traskov2012scheduling, jones2012optimal, cui2007distributed, wiese2016scheduling, niati2012throughput} for various objectives, e.g., maximizing throughput or minimizing the energy consumption under certain constraints. 
These approaches are either converging-to-optimal with respect to some constraints or only approximate the solutions.

The closest work to ours is~\cite{traskov2012scheduling}, where the authors decompose the joint optimization of scheduling and network coding into two subproblems, similar with us decomposing the joint MMF (and MCMF) and scheduling problems. 
However, some differences are as follows. 
Except we provide a unified framework that also covers the case where the propagation delays are non-negligible and utilized in scheduling (see Section~\ref{subsec::review_delay}), we study the multi-source (instead of single source) multi-sink case, where the trade-off between the rates of the sources becomes an important factor of consideration. 
Moreover, the algorithm in \cite{traskov2012scheduling} is an iterative algorithm that only converges to the optimum, but our algorithms provably output the exact optimum in a finite number of iterations.

\subsection{Networks with Non-negligible Propagation Delays} 
\label{subsec::review_delay}

In the existing theory of terrestrial wireless communications, though the communication media (e.g., radio, light, sound) have nonzero propagation delays, they are usually short and treated as a factor of interference~\cite{proakis2003intersymbol}. 
However, in certain environments, e.g., deep-space and underwater networks, the propagation delays can be significantly longer. 
For example, for sound (whose speed is about $1.5$ kilometers per second) to propagate over a distance of $3$ kilometers, the delay can be about $2$ seconds. 
Recent studies~\cite{ma2021rate, yang2023wireless_tit, fan22isit,  hsu2009st, guan2011mac, chitre2012throughput, bai2017throughput} show that it is promising to \emph{utilize} the delays to significantly improve the network performance. 
For such scenarios, mixed integer linear programming for some heuristic algorithms~\cite{hsu2009st, guan2011mac} and dynamic-programming based algorithms~\cite{chitre2012throughput} were proposed.

In~\cite{yang2023wireless_tit, ma2021rate}, by extending the conflict graph~\cite{jain2005impact} to a \emph{weighted} graph, the scheduling rates can be exactly characterized with even higher complexity. 
We extend their formulation in our framework, and provide an efficient algorithm that solves the multiflow maximization problems while also utilizing the delays to improve the scheduling rates. 
We note that existing methods for the multiflow problems cannot be directly extended to this case.

\section{Network Model}
\label{sec::no_delay_net}
In this section, we define the network model and formulate the MMF and MCMF problems. 
We note that we include propagation delays in the network model for the purpose of generalizing our framework in Section~\ref{sec::delay_net}, but we will first focus on the algorithm for networks without considering delays, as our main contributions lie in the joint algorithm and its generality and optimality.

\subsection{Network Model}
We use a link-wise network model~\cite{ma2021rate, yang2023wireless_tit, wan2009multiflows, jain2005impact}, where each link is associated with a \emph{collision set}, including all the links that can be interfered by it. 
It is called the binary interference model~\cite{jain2005impact, ma2021rate}, and it is not difficult to cover the physical interference model by signal-to-interference-and-noise ratio~\cite{wiese2016scheduling, yang2023wireless_tit}. 
We assume the network is acyclic and \emph{discrete}~\cite{ma2021rate, fan22isit, yang2023wireless_tit} in the sense that time is slotted and the link delays are multiples of a length of a time slot, which is justified in~\cite{fan22isit}. 
The intermediate nodes can wait until enough packets are collected before performing coding on the packet.

The network can be modeled by a tuple $\mc N=(\mc V, \mc L,\mc I, D)$, where $\mc V$ is the \emph{node set},
$\mc L \subseteq {\mc V}^2$ is the \emph{link set}, $\mc I=(\mc I(l), l\in\mc L)$ is the set of  \emph{collision sets} where $l'\in\mc I(l)$ if $l'$ is in the interference range of $l$
and $D : {\mc L}^2 \to \mathbb{Z}$ is the link-wise delay matrix specifies the delays between links. 
We assume each link has a unit bandwidth and allow parallel links between nodes. 
$(\mc L,\mc I, D)$ can form a weighted, directed graph $\mc N$ where $\mc L$ is the finite vertex set, $(l,l')$ is an edge if $l'\in \mc I(l)$, and $D(l,l')$ is the weight on the directed edge $(l,l')$, which degrades to an unweighted graph $(\mc L, \mc I)$ if delays are ignored~\cite{wan2009multiflows, zhou2013maximum}. 
This graphical approach helps the discussion on our algorithms.

We now describe the communication task over the network $\mc N=(\mc V, \mc L,\mc I, D)$. 
For $k\in\mathbb{N}_+$, let $\mc S = \{s_1, \ldots, s_{k}\}\subseteq \mc V$ be the set of source nodes. 
We assume the information sources at different source nodes are mutually independent. 
Each source node $s_i$ is associated with a set of sink nodes $\mc D_{s_i}\subseteq \mc V$ that have to decode the information at $s_i$. 
We allow a node to be both a source node and a sink node corresponding to another source node. 
For $i\neq j$, we may have $\mc D_{s_i}\cap \mc D_{s_j}\neq \emptyset$, i.e., different source nodes can share same sink nodes. 
Each link $l\in\mc L$ represents a point-to-point channel with unit capacity. 
The sets of input channels and output channels of a node $v\in\mc V$ are denoted
by $\text{In}(v) \subseteq {\mc L}$ and $\text{Out}(v) \subseteq {\mc L}$, respectively.

To explain the model, we use a line network~\cite{ma2021rate, yang2023wireless_tit, chitre2012throughput, bai2017throughput} as an example.

\begin{example}
\label{eg::N41_noDelay}

Consider an $L$-hop unicast line network: there are $L+1$ nodes $\mc V = \{1, 2, \ldots, L+1\}$, with link set
\begin{equation*}
\mc L = \{l_i\triangleq (i,i+1), i = 1,\ldots,L\}. 
\end{equation*}
Each link has a unit delay. 
We consider a $K$-hop interference model: the reception of a node has possible collisions from nodes in $K$ hops. 
The nodes are half-duplex. 
The collision set of $l_i$ is
\begin{equation}
\label{eq:exi}
\mc I_K(l_i) = \{l_j: j\neq i, |i+1-j| \leq  K\}. 
\end{equation}
The link $l_i$ is active in time slot $t$ if node $i$ sends a signal in time slot $t$ to node $i+1$. 
Hence the link-wise delay matrix is 
\begin{equation}
\label{eq:exd}
D(l_i,l_j) =1 - |i+1-j|. 
\end{equation}
We denote it by $\mc N_{L,K}^{D=1}$, and it can be represented by a graph, where $\mc N_{4,1}^{D=1}$ is shown in Figure~\ref{fig:weighted_N41} as an example. 
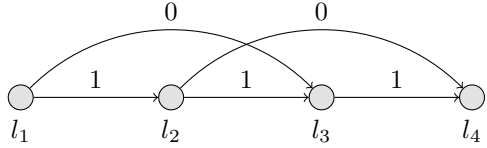
\begin{figure}
\centering
\begin{tikzpicture}[link/.style={circle, draw, thin,fill=darkgray!15}]
\foreach \place/\x in {{(-6,0)/1},{(-4,0)/2}, {(-2,0)/3},{(0,0)/4}}
\node[link,label=below:$l_\x$] (u\x) at \place {};

\draw[->] (u1) to node[above] {$1$} (u2);
\draw[->] (u2) to node[above] {$1$} (u3);
\draw[->] (u3) to node[above] {$1$} (u4);

\draw[->, bend left=45] (u1) to node[above] {$0$} (u3);
\draw[->, bend left=45] (u2) to node[above] {$0$} (u4); 

\end{tikzpicture}
\vspace{-11pt}
\caption{$\mc N_{4,1}^{D=1}$: nodes represents network link, edges represent the collision relations between nodes, and edge weights represent propagation delays. 
}
\label{fig:weighted_N41}
\end{figure}
\end{example}

\subsection{Collision-free Schedules and Rate Region}

We define collision-free schedules, similar to ~\cite{ma2021rate,yang2023wireless_tit}. 
Since we assume the time is slotted, when link $l$ is active at time slot $t$ and link $l'\in \mc I(l)$ is active at time slot $t+D(l,l')$, we say a \emph{collision occurs}. 
In each time slot we use a binary number to indicate whether a link sends messages or not. 
Hence we use an infinite binary matrix $S : \mc L \times \mathbb{N} \to \{0,1\}$ with rows indexed by $\mc L$ and columns indexed by $\mathbb{N}$ to specify a \emph{schedule}: $S(l,t)=1$ indicates that link $l$ is active in time slot $t$, and $S(l,t)=0$ indicates it is inactive. 
$S(l,t)$ has a collision in $\mc N$ if $S(l,t)=S(l',t+D(l,l'))=1$ for a certain $l'\in\mc I(l)$. Otherwise $S(l,t)$ is \emph{collision free}. 
A schedule $S$ is collision free if $S(l,t)$ is collision free for all $(l,t)$. 
There are different (though with similar ideas) definitions if delays are simply ignored, e.g., \cite{wan2009multiflows}, but we aim to provide a framework general enough to cover the networks with non-negligible delays.  
For a collision-free schedule $S$ and a link $l$, the link rate is
\begin{equation}
\label{eq::rate_vectors}
R_{S}(l) = \lim_{T\rightarrow \infty} \frac{1}{T} \sum_{t=0}^{T-1} S(l,t).
\end{equation}
If $R_S(l)$ exists for all $l\in \mc L$, we call $R_S = (R_S(l),l\in \mc L)$ the \emph{rate vector} of $S$.  
For a network $\mc N=(\mc V , \mc L, \mc I, D)$, a rate vector $R=(R(l),l\in \mc L)$ is said to be \emph{achievable} if for all $\epsilon>0$, there exists a collision-free schedule $S$ such that
$R_S(l)> R(l)-\epsilon$  for all $l\in\mc L$. 
For a link $l$, the rate $R(l)$ can stand for the maximum number of information symbols that can be sent on the channel per time slot. 
Then each achievable rate vector can be viewed as a rate constraint for the network.

The collection $\dR(\mc N)$ of all the achievable rate vectors is called the \emph{(scheduling) rate region} of $\mathcal{N}$. 
We may use $\dR$ instead of $\dR(\mc N)$ to simplify the notation when the context is clear. 
It is proved in~\cite{ma2021rate} that $\dR$ is a convex polytope, and can be achieved by using periodic collision-free schedules.

\subsection{Problem Formulation}
\label{subsec::NC}

We define the maximum multiflow (MMF) and the maximum concurrent multiflow (MCMF) problems. 
Most existing literature consider \emph{multiple unicast}, i.e., each source only has one corresponding sink with a certain demand~\cite{su2008characterizing, wan2009multiflows, zhou2013maximum}. 
However, we discuss general \emph{multiple multicast}, with network coding~\cite{ahlswede2000network, li2003linear, yeung2008information} on intermediate nodes, where the nodes can encode their received data before passing them on.

For each multicast session, one source corresponds to multiple sinks. 
Though we use network coding to attain the maximum information flow in a session, we do not consider coding \emph{between} sessions (i.e., inter-session network coding~\cite{eryilmaz2007control, yang2008linear}) in this paper for the sake of simplicity, since it is in general a hard problem~\cite{langberg2011hardness} and can even be undecidable~\cite{li2023undecidability}. We say 
\begin{equation}
\label{eq::flow}
F = (F(l) \in\mathbb{N}_{\geq 0} : \, l\in\mc L)
\end{equation}
is a valid flow
from source node $s \in \mc V$ to sink node $t \in \mc V$ with respect to a rate vector $R$ 
if it satisfies: 
\begin{itemize}
\item $0\leq F(l)\leq R(l)$ for all $l \in \mc L$, i.e., the flow along link $l$ does not exceed the rate constraint $R(l)$. 
\item The flow conservation equation $\sum_{l\in \text{In}(v)}F(l) =  \sum_{l\in \text{Out}(v)}F(l)$ for all $v\in\mc V \backslash \{s,t\}$. 
\end{itemize}

We see the flow $ \sum_{l\in \text{Out}(s)}F(l)$ out of $s$ equals to the flow $\sum_{l\in \text{In}(t)}F(l)$ into $t$, and this value is called the \emph{value} of $F$, denoted as $\mathrm{val}(F)$. 
We say $F$ is a \emph{max-flow} from $s$ to $t$ with respect to $R$ if $F$ is a flow and has a value no smaller than the value of any other flow from $s$ to $t$ with respect to $R$.

\subsubsection{Maximum Multiflow (MMF) Problem}
\label{subsec::Multiflow}

Consider a source $s_i$ multicasts a message to nodes in the set $\mc D_{s_i} = \{t_{i,1},\ldots ,t_{i,k_i}\}$, with network coding~\cite{ahlswede2000network, li2003linear, yeung2008information}, at a rate
\[
v_i = \min_{j} \mathrm{val}(F_{i,j}),
\]
where $F_{i,j}$ is a flow from $s_i$ to $t_{i,j}$ for each $j$ and the rate of communication along link $l$ is $\max_{j}F_{i,j}(l)$.

We now put the flows from the source nodes $s_1, \ldots, s_{k}$ together. 
Fix a rate vector $R$. 
Link $l$ has to accommodate all these $k$ flow requirements simultaneously, i.e., $\sum_{i=1}^k \max_{j}F_{i,j}(l) \le R(l)$ for all $l \in \mc L$. 
We maximize the sum of the rates of multicasting these $k$ sources, and it is called the maximum multiflow (MMF) problem, which is formulated by the following linear program, combining the linear program for multiple unicast~\cite{jain2005impact, wan2009multiflows, zhou2013maximum} and the program for single multicast~\cite{traskov2012scheduling}: 
\begin{align}
&  \text{LP-MMF}(\tilde{\mathcal{R}}): \nonumber\\ 
&\quad \text{maximize}\quad \sum_{i=1}^k v_i \\
& \quad \text{subject to} \nonumber\\
&\quad \, F_{i,j} \text{ is a valid flow},\, \forall  i\in [k], j\in [k_i], \nonumber\\ 
&\quad \,\sum_{l \in \mathrm{Out}(s_i)}F_{i,j}(l) = \sum_{l \in \mathrm{In}(t_{i,j})}F_{i,j}(l) = v_i,\, \forall  i\in [k], j\in [k_i], \nonumber\\
&\quad \,G_{i}(l) \geq F_{i,j}(l), \,\forall l\in \mc L, i\in [k], j\in [k_i], \nonumber\\
&\quad \, \sum_{i=1}^k G_i(l) \leq  R(l), \,\forall l\in \mc L, \label{eq:dual}\\
&\quad \, R \in \tilde{\mathcal{R}}, \nonumber
\end{align}
where $[k]:= \{1,\ldots,k \}$, $k, k_i\in\mathbb{N}_+, \forall i$. 
The linear program takes a polytope $\tilde{\mathcal{R}}$ (a subset of the scheduling rate region) as an input.
The variables $G_i(l)$, $i\in [k]$, $l \in \mc L$ are used to impose the constraint $\sum_{i=1}^k \max_{j}F_{i,j}(l) \le R(l)$. 
The constraint~\eqref{eq:dual} gives a dual vector that will be used in our algorithms, and the dual variable corresponding to link $l$ means how sensitive the optimization objective is to the rate constraint $R(l)$.

\subsubsection{Maximum Concurrent Multiflow (MCMF) Problem}
While the MMF problem is to find the link schedule that can support the maximum total 
rate of transmission of the sources, the maximum concurrent multiflow (MCMF) problem is to 
find the link schedule 
such that all the sources can transmit concurrently at the maximum rate~\cite{shahrokhi1990maximum, kodialam2003characterizing, wan2009multiflows}. 
The settings of MCMF problems in~\cite{kodialam2003characterizing, wan2009multiflows} are also for multiple unicast.

More generally, instead of maximizing the sum rate $\sum_{i=1}^k v_i$, we maximize $\phi$ such that the source $s_i$ can multicast at a rate $v_i = \phi \gamma_i$, where $\gamma_i$ is the desired traffic rate at $s_i$.
The MCMF problem is formulated as follows~\cite{wan2009multiflows,traskov2012scheduling}: 
\begin{align}
& \text{LP-MCMF}(\tilde{\mathcal{R}}): \nonumber\\ 
&\quad  \text{maximize}\quad \phi \\
& \quad \text{subject to} \nonumber\\
&\quad\, F_{i,j} \text{ is a valid flow},\, \forall  i\in [k], j\in [k_i], \nonumber\\ 
& \quad \sum_{l \in \mathrm{Out}(s_i)}F_{i,j}(l) = \sum_{l \in \mathrm{In}(t_{i,j})}F_{i,j}(l) = \phi \gamma_i,\, \forall  i\in [k], j\in [k_i], \nonumber\\
&\quad\,  G_{i}(l) \geq F_{i,j}(l), \,\forall l\in \mc L, i\in [k], j\in [k_i], \nonumber\\
&\quad \, \sum_{i=1}^k G_i(l) \leq  R(l), \,\forall l\in \mc L, \label{eq:dual_mcmf} \\
&\quad \, R \in \tilde{\mathcal{R}}. \nonumber
\end{align}
We also utilize the dual vector given by \eqref{eq:dual_mcmf} in the algorithms.

\section{Algorithm}
\label{sec::alg}

In this section, we present our main algorithm that jointly computes the MMF (or MCMF) and the scheduling rate region, thereby provably outputting the exact-optimal solution while requiring only a subset of the scheduling rate region. 
Note that in this section, we focus on networks without considering delays, as the algorithm, along with its generality and optimality, constitutes our main contributions.

Since a collision-free schedule can be found by searching an independent set in the graph $(\mc L, \mc I)$~\cite{ma2021rate}, we attach a weight $a_i \ge 0$ to link $l_i$, and would like to maximize the weighted total rate, i.e., for the scheduling rate region $\mc R$, we solve
\begin{equation}
\label{eq::inner_prod}
\argmax_{R\in \mc R} \langle\mathbf{a}, R\rangle,
\end{equation}
where $\mathbf{a} = (a_i)_{i=1,\ldots,|\mc L|}$, and $\langle\cdot,\cdot \rangle$ is the inner product.

\begin{remark}
    For~\eqref{eq::inner_prod}, the objective is to maximize the weighted sum rate instead of just the sum rate, since a different weight vector could be used in each iteration of our algorithms (see step $4$ of Algorithm~\ref{alg::sche_NC_1} or step $5$ of Algorithm~\ref{alg::sche_NC_2}). 
    These weights are also crucial for the graphical approach we propose for Algorithm~\ref{alg::sche_NC_2}. 
\end{remark}

This corresponds to a weighted maximal independent set problem~\cite{wan2009multiflows,traskov2012scheduling} that can be solved by integer linear programming (ILP) by maximizing over $S(l_i) \in \{0,1\}$ for $l_i\in\mc L$: 
\begin{align*}
\text{ILP: }\quad \text{maximize}\quad & \sum_{i=1}^{|\mc L|} a_i S(l_i) \\
\text{subject to}\quad & S(l_i) + S(l_j)\leq 1,  \quad \forall\, l_i,l_j:\, l_j\in \mc I(l_i).
\end{align*}
The solution gives us a maximal independent set of $(\mc L, \mc I)$, and the corresponding achievable rate vector is $S=(S(l_i), l_i\in\mc L)$, which is a vertex of the scheduling rate region $\mc R$.

Based on~\eqref{eq::inner_prod}, we iteratively search the MMF or MCMF and the scheduling region. 
Even though our target is the optimal value instead of approximated or converging-to-optimal values, we show that it is unnecessary to find the entire scheduling region before solving the MMF or MCMF problem.

Suppose $\mathrm{flow}(\cdot)$ is the function for calculating the MMF or MCMF in a given polytope, which can be a subset of the scheduling region. 
From $i=1$, in each iteration, the algorithm works as follows: 
\begin{enumerate}
\item We start with a subset of the scheduling region $\mc R_i$, which is formed by the vertices of $\mc R$ we have known (it is reasonable to assume some rate vectors are known, e.g., by activating the first link all the time and inactivating others, the vector $[1,0,\ldots,0]^\intercal$ is achievable). 
In the first iteration, we start with an arbitrarily chosen rate vector $R_1$, i.e, $\mc R_1 =\{R_1\} $. 
We run the linear program LP-MMF$(\mc R_i)$ (or LP-MCMF$(\mc R_i)$) to obtain the dual vector $\mu_i$, corresponding to the constraint in~\eqref{eq:dual} (or~\eqref{eq:dual_mcmf}).

\item By the ILP, we find a new rate vector $R_{i+1}$ by
\begin{equation}
    R_{i+1}=\argmax_{R\in\mc R} \langle\mu_i, R\rangle . 
    \label{eq:Ri_argmax}
\end{equation}  

\item We update the subset of the scheduling region by computing the convex hull $\mc R_{i+1} =\conv \left(\mc R_i\cup \{R_{i+1}\}\right) $.

\item If $\langle 
\mu_{i}, R_{i+1}
\rangle = \max_{R\in \mc R_{i}}
\langle 
\mu_{i}, R
\rangle $, the algorithm terminates and outputs the last optimal value of LP-MMF$(\mc R_i)$ (or LP-MCMF$(\mc R_i)$); otherwise it comes back to step $1$ and continues. 
\end{enumerate}

We then prove that Algorithm~\ref{alg::sche_NC_1} will terminate and output the maximum (concurrent) multiflow in finite iterations.

 \begin{algorithm}[tb]
\caption{}
\label{alg::sche_NC_1}
\textbf{Input}: a network $(\mc V, \mc L,\mc I)$

\textbf{Output}: maximum multiflow $v$	
\begin{algorithmic}[1]
\State Start with any rate vector $R_1\in \mc R$, $\mathcal{R}_1 \leftarrow \{R_1\}$
\For{$i=1,2,\ldots$}
\State Run $v_i \leftarrow$LP-MMF$(\mc R_i)$ (or $v_i \leftarrow$LP-MCMF$(\mc R_i)$) and obtain the dual vector $\mu_i$
\State Run ILP to find $R_{i+1} \leftarrow \argmax\limits_{R\in\mc R} \langle\mu_i, R\rangle $
\State $\mc R_{i+1} \leftarrow \conv \left(\mc R_i\cup \{R_{i+1}\}\right) $
\If{$\langle 
\mu_{i}, R_{i+1}
\rangle = \max_{R\in \mc R_{i}}
\langle 
\mu_{i}, R
\rangle $}\\
$\qquad $
\Return $v_i$
\EndIf
\EndFor
\end{algorithmic}
\end{algorithm}

\begin{theorem}[Optimality]
\label{thm::optimality_alg1}
For a network $\mc N=(\mc V, \mc L,\mc I)$, Algorithm~\ref{alg::sche_NC_1} will terminate and output the maximum multiflow (or the maximum concurrent multiflow). 
\end{theorem}
\begin{proof}
The scheduling rate region is a polytope~\cite{jain2005impact} with a finite number of vertices, and hence Algorithm~\ref{alg::sche_NC_1} will eventually terminate, since the worst case is that all the vertices are found to solve the maximum (concurrent) multiflow problem.

We then show that the output is optimal; that is, when the algorithm terminates, the output is exactly the MMF or MCMF. 
We use $\mathrm{flow}(R)$ to denote the function that outputs the MMF or MCMF with respect to a rate vector $R$.

Denote $\mc R$ as the entire rate region (which may not need to be found), and $\mc R_i$ is the subset of $\mc R$ found until iteration $i$. 
Suppose the algorithm terminates at iteration $i'$, i.e., $\langle \mu_{i'}, R_{i'+1} \rangle = \max_{R\in \mc R_{i'}} \langle \mu_{i'}, R \rangle $. By substituting \eqref{eq:Ri_argmax}, 
\begin{align}
    \max_{R\in \mc R} \langle \mu_{i'}, R \rangle = \max_{R\in \mc R_{i'}} \langle \mu_{i'}, R \rangle. \label{eq:max_same}
\end{align}

Suppose the optimal $R$ in LP-MMF$(\mc R_{i'})$ (or LP-MCMF$(\mc R_{i'})$) is $R^*$.
Note we can write LP-MMF$(\mc R_{i'})$ as
\begin{align}
\text{maximize} & \;\;\mathrm{flow}(R) - \chi_{\mc R_{i'}}(\bar{R}) \nonumber\\
\text{subject to} & \;\; R=\bar{R}, \label{eq:lp_simplified}
\end{align}
where $\chi_{\mc R_{i'}}(\bar{R})$ is the characteristic function (which is $0$ if $\bar{R} \in \mc R_{i'}$, or $\infty$ otherwise), which forces $\bar{R} \in \mc R_{i'}$. The dual vector $\mu_{i'}$ corresponding to the constraint in~\eqref{eq:dual} (or~\eqref{eq:dual_mcmf}) is the same as the dual vector corresponding to the constraint $R=\bar{R}$ in \eqref{eq:lp_simplified}. Considering the Lagrangian of \eqref{eq:lp_simplified}, at the optimum $(R^*,\bar{R}^*)$, the subgradient satisfies $0 \in \partial_{R^*} \mathrm{flow}(R^*) - \partial_{R^*} \chi_{\mc R_{i'}}(\bar{R}^*) - \mu_{i'}$ and $0 \in \partial_{\bar{R}^*} \mathrm{flow}(R^*) - \partial_{\bar{R}^*} \chi_{\mc R_{i'}}(\bar{R}^*) + \mu_{i'}$, and we have $R^* = \bar{R}^*$.
Hence, $R^*$ maximizes $\mathrm{flow}(R) -\langle \mu_{i'}, R \rangle$ for $R\in \mathbb{R}^{|\mathcal{L}|}_{\ge 0}$, and maximizes $\langle \mu_{i'}, R \rangle$ for $R\in \mc R_{i'}$. 
%Our goal is to argue that $R^*$ also maximizes $\mathrm{flow}(R)$. 

Fix any $R'\in \mc R$.
Since $R^*$ maximizes $\mathrm{flow}(R) -\langle \mu_{i'}, R \rangle$ for $R\in \mathbb{R}^{|\mathcal{L}|}_{\ge 0}$, 
\begin{align*}
\label{eq::alg_proof_2}
& \mathrm{flow}(R^*) - \langle \mu_{i'}, R^* \rangle \\
&\geq \mathrm{flow}(R') - \langle \mu_{i'}, R' \rangle \\
& \ge \mathrm{flow}(R') - \langle \mu_{i'}, R^* \rangle,
\end{align*}
where the last inequality is by \eqref{eq:max_same} and the fact that $R^*$ maximizes $\langle \mu_{i'}, R \rangle$ for $R\in \mc R_{i'}$.
Therefore, $R^*$ maximizes $\mathrm{flow}(R)$ for $R\in \mc R$ and the proof is finished. 
\end{proof}

\begin{remark}     
\label{rmk::ILP}
Algorithm~\ref{alg::sche_NC_1} requires integer linear programming, hence it does not have a polynomial time complexity, which is expected since this problem is NP-hard~\cite{jain2005impact, wan2009multiflows}. 
The efficiency will be verified via simulation results in Section~\ref{sec::simu}. 
Algorithm~\ref{alg::sche_NC_1} and Theorem~\ref{thm::optimality_alg1} pave the way to the Algorithm~\ref{alg::sche_NC_2} (and its optimality) for networks with non-negligible propagation delays, see Section~\ref{sec::delay_net} for a detailed discussion. 
\end{remark}

\begin{remark} 
\label{remark2}
In~\cite{traskov2012scheduling}, an algorithm based on subgradient optimization that decomposes the problem into two parts has been discussed. 
Though it shares some similarities with ours, our algorithm is guaranteed to find the optimum exactly in a finite number of steps (assuming access to an integer linear programming algorithm), whereas~\cite{traskov2012scheduling} is an iterative algorithm that only converges to the optimum. 
As we will see in the following sections, terminating in a small number of steps is especially important for networks with non-negligible delays, since the update of the subset $\mc R_i$ of the scheduling region is the bottleneck of the algorithm with exponential time complexity, and should be performed as little as possible. 
\end{remark}

\begin{example}
\label{eg::2_line}
Suppose we have a $3$-node, $2$-link unicast network, denoted by $\mc N_{2,1}^{D=0}$, under the $1$-hop interference model, and the propagation delay on each link is negligible (i.e., zero). 
Starting from $R_1=\begin{bmatrix}
1&0
\end{bmatrix}^\intercal$, 
we solve 
\begin{align*}
\max  \quad& v =R(l_1)=R(l_2)\\
\text{s.t. } &   R(l_1) \leq 1,  \\
&   R(l_2) \leq 0,  \\
&    R(l_1), R(l_2) \geq 0,
\end{align*} 
which gives us a maximum flow of value $v=0$ in $\mc R_i = \{R_1\}$. 
We then use the ILP to find another rate vector: 
\begin{align*}
\max\quad & R(l_1) + R(l_2) \\
\mathrm{s.t.}\quad  & R(l_1)+R(l_2)\leq 1,
\end{align*}
which gives $R_2=\begin{bmatrix}
0&1
\end{bmatrix}^\intercal$ and $\mc R_2 = \conv(R_1, R_2)$. 
Next, 
\begin{align*}
\max  \quad& v=R(l_1)=R(l_2)\\
\text{s.t. } & R(l_1)+R(l_2)\leq 1 \\
 &R(l_1),R(l_2) \geq 0
\end{align*} 
gives us the maximum flow in $\mc R_2$,  $ v=1/2$. 
We can verify that the resulting dual vector $\mu_2$ in this iteration gives 
\begin{equation*}
f(\mc R_2, \mu_2)=  \argmax\limits_{R\in\mc R_2}  f(\mc R_2, \mu_2),
\end{equation*}
which meets the condition that the algorithm terminates. 
It is also straightforward  to verify that $1/2$ is indeed the maximum flow value. 
\end{example}

\section{Networks with Non-negligible Delays}
\label{sec::delay_net}

In this section, we extend our previous discussions to a special case: networks with non-negligible propagation delays (e.g., underwater or deep-space environments). 
Recent studies~\cite{chitre2012throughput, ma2021rate, yang2023wireless_tit, bai2017throughput, hsu2009st} have shown that, for such networks, instead of using guard intervals to mitigate the effects of delays, we can actually \emph{utilize} the propagation delays to significantly improve the scheduling rate region.

We adopt a graphical approach, build upon the framework proposed by~\cite{ma2021rate, yang2023wireless_tit}, and provide an algorithm in the same spirit as Algorithm~\ref{alg::sche_NC_1}. 
The proposed Algorithm~\ref{alg::sche_NC_2} also provably outputs exact-optimal solutions in a finite number of iterations, and maintains all the advantages (generality, efficiency, and optimality).

The key to solve the MMF or MCMF problems in the networks with non-negligible delays is, we need a function similar to~\eqref{eq::inner_prod} that can output a vertex of the scheduling region in a time complexity at most exponential in $|\mathcal{L}|$ (which in turn will be polynomial in the size of the \emph{scheduling graphs} below), which is more efficient than the cycle-enumeration approach~\cite{ma2021rate} with complexity doubly exponential in $|\mathcal{L}|$.

We review the \emph{scheduling graph} in~\cite{ma2021rate} as follows, which will be generalized later: 
For a collision-free schedule matrix $S$ and integers $T\in \mathbb{N}_{+}$, $k\in \mathbb{Z}$, 
$S[T,k]$ denotes the submatrix of $S$ consisting of columns  $kT,kT+1,\ldots,(k+1)T-1$. 
If a submatrix $S'$ is formed by $T$ columns of  $S$, its columns are indexed by $0,1,\ldots,T-1$.

\begin{definition}[Scheduling Graph~\cite{ma2021rate}]
\label{def:schedule_graph}
Given a network $\mc N$ and an integer $T>0$, 
the \emph{scheduling graph}  $(\mc M_{T}, \mc E_{T})$ is a directed graph that is defined as follows: 
the vertex set $\mc M_{T}$ includes all the $|\mc L|\times T$ binary matrices $A$ such that $A = S[T,0]$ for a certain collision-free schedule $S$. 
The edge set $\mc E_{T}$ includes all the vertex pairs $(A,B)$ such that $A = S[T,0]$ and $B = S[T,1]$ for a certain collision-free schedule $S$. 
\end{definition}

\begin{remark}
    For a network with negligible propagation delays, the scheduling graph $(\mc M_{T}, \mc E_{T})$ is a complete graph. 
\end{remark}

In~\cite{ma2021rate}, it has been shown that by choosing $T\geq \max_{l\in \mc L} \max_{l'\in \mc I(l)}|D(l,l')| $, calculating the scheduling region is equivalent to searching all the simple cycles in the scheduling graph, which is NP-hard. 
The scheduling problem then may then even have doubly exponential complexity, since the number of vertices in $(\mc M_{T}, \mc E_{T})$ increases exponentially with respect to $|\mc L|$, and the cycle enumeration in $(\mc M_{T}, \mc E_{T})$ is also NP-hard in general.

Instead of enumerating cycles for the scheduling region, we search \emph{maximum-mean-cycles} (which can be solved in a time complexity \emph{polynomial} in the graph size) in a new graph, to find the vertices of the scheduling region. 
We may only need a few rate vectors to solve the MMF or MCMF problem.

Before describing our approach, we need some graphical concepts. 
In a directed graph $\mc G$, a \emph{path} is a sequence of vertices $v_0,v_1 \ldots, v_m$ where for $i=0,1,\ldots,m-1$, $(v_i, v_{i+1})$ is a directed edge. 
A path is \emph{closed} if $v_0=v_m$. 
A \emph{cycle} in $\mc G$ is a closed path $(v_0,v_1 \ldots, v_m)$ such that $m\geq 1$, $v_i\neq v_j$ for any $0\leq i\neq j\leq m-1$ and $v_0=v_m$, i.e., in such a sequence, the only repeated vertices are the first and the last vertices. 
Note a closed path can be decomposed into  a sequence of cycles~\cite{gleiss2001circuit}, and this has been used in proving that it suffices to enumerate all the simple cycles for calculating the scheduling rate region~\cite{ma2021rate}. 
In a graph where each edge is associated with a weight, we say the weight of a directed cycle is the total weight on the edges in the cycle. 
Then we say the average weight of a directed cycle is the total weight divided by the number of edges in the cycle. 
The \emph{maximum-mean-cycle} is the cycle in the given weighted, directed graph with the maximum average weight over all directed cycles in the given graph.

It has been proved in~\cite{ma2021rate, yang2023wireless_tit} that a collision-free, periodic schedule is equivalent to a closed path (which can be decomposed to multiple simple cycles) in $(\mc M_T, \mc E_T)$ and vice versa, i.e., the concatenation of a sequence of vertices in $(\mc M_T, \mc E_T)$ (which are matrices of size $|\mc L|\times T$) forms a periodic, collision-free schedule. 
In this paper, 
we define a \emph{weighted scheduling graph} and use the maximum-mean-cycle in it to solve~\eqref{eq::inner_prod}.

\begin{definition}[Weighted Scheduling Graph]
\label{def::weighted_SG}
Given a weight vector $\mathbf{a}\in \mathbb{R}^{|\mathcal{L}|}$ and a scheduling graph $(\mc M_T, \mc E_T)$ whose vertices are matrices of size $|\mc L|\times T$, a \emph{weighted scheduling graph} $(\mc M_T, \mc E_T, w_\mathbf{a})$ is a directed, weighted graph defined as follows: the vertex set is still $\mc M_T$, and each edge is associated with a weight. 
For a directed edge $(v_1, v_2)$ in $(\mc M_T, \mc E_T)$, there is a weighed, directed edge $(v_1, v_2)$ in $(\mc M_T, \mc E_T, w_\mathbf{a})$ with weight $w_\mathbf{a}(v_1, v_2) =\mathbf{a}^\intercal v_2 \mathbf{1}$, where $\mathbf{1} = [1,\ldots,1]^\intercal \in \mathbb{R}^{T}$. 
\end{definition}

Since each achievable rate vector can be achieved by a periodic, collision-free schedule, which corresponds to a cycle in  $(\mc M_T, \mc E_T)$~\cite{ma2021rate}, we have the following result.

\begin{lemma}
\label{lemma::weighted_SG}
    For a weighted scheduling graph $(\mc M_T, \mc E_T, w_\mathbf{a})$ and its maximum-mean-cycle  $\mc C = (v_0, v_1, \ldots, v_m)$ with $m\geq 0$ and $v_0=v_m$, 
    the concatenation of the vertices in $\mc C$ gives a periodic schedule $S'$ such that $R_{S'} \in \argmax_{R\in \mc R} \langle\mathbf{a}, R\rangle$. 
\end{lemma}

\begin{proof}
    For any rate vector $R\in \mc R$, suppose it is achieved by a schedule $S$ that corresponds to a cycle $ (v_0, v_1, \ldots, v_m)$ in $(\mc M_T, \mc E_T)$, 
    by the definition of rate vectors~\eqref{eq::rate_vectors}. Given $\mathbf{a} = (a_i)_{i=1,\ldots,|\mc L|}$, 
    \begin{align*}
        \langle\mathbf{a}, R\rangle 
        & = \sum_{i=1}^{|\mc L|} a_i  R(l_i) \\
        & =  \sum_{i=1}^{|\mc L|} a_i \cdot \frac{1}{mT} \sum_{t=0}^{mT-1} S(l_i, t)\\
        &=   \frac{1}{mT} \sum_{i=1}^{|\mc L|}  \sum_{t=0}^{mT-1} a_i \cdot S(l_i, t)\\
        &= \sum_{j=0}^{m-1}\left( \frac{1}{mT}  \sum_{i=1}^{|\mc L|} \sum_{k=0}^{T-1} a_i \left(v_j(i,k)\right)
        \right) \\
        &= \frac{1}{mT} \sum_{j=0}^{m-1} \mathbf{a}^\intercal v_j \mathbf{1},
    \end{align*}
    which states that $\langle\mathbf{a}, R\rangle$ equals to the average over values $\mathbf{a}^\intercal v_j \mathbf{1}$ for $j=0,1,\ldots,m-1$. 
    % Therefore, it justifies the construction of $(\mc M_T, \mc E_T, w_\mathbf{a})$ and shows that it is indeed the maximum-mean-cycle that attains the maximum of $\langle\mathbf{a}, R\rangle$. 
\end{proof}

Therefore, given a vector $\mathbf{a}$, finding a vector that solves~\eqref{eq::inner_prod} is equivalent to finding the maximum-mean-cycle in $(\mc M_T, \mc E_T, w_\mathbf{a})$, which is a widely studied problem~\cite{karp1978characterization, dasdan1998faster} that can be solved with time complexity $\Theta(nm)$, where $n$ is the number of nodes and $m$ is the number of edges in the graph. 
A classical algorithm is the Karp's algorithm~\cite{karp1978characterization}, which is briefly reviewed below for the sake of completeness.

Given the graph with vertex set $\mc V$ and a source node $s\in \mc V$, for each $v\in \mc V$ and every non-negative integer $k$, suppose $F_k(v)$ denotes the maximum weight of a length-$k$ path from $s$ to $v$, and we say $F_k(v)=-\infty$ if such a path does not exist. 
Then the maximum cycle mean $\lambda^*$ can be derived from the following theorem, whose proof can be found in~\cite{karp1978characterization}.

Note we assume the graph $(\mc M_T, \mc E_T)$ is strongly connected, and hence $(\mc M_T, \mc E_T, w_\mathbf{a})$ is also strongly connected. 
Otherwise, we find the strongly connected components (with linear time complexity), search for the maximum-mean-cycle in each component and choose the one with the largest cycle mean.

\begin{theorem}[Karp's Theorem~\cite{karp1978characterization}]
Given a strongly connected graph, the maximum cycle mean $\lambda^*$ is given by
\begin{equation}
\lambda^* = \max_{v\in \mc V} \min_{0\leq k\leq n-1} \frac{F_n(v) - F_k(v)}{n-k},
\end{equation}
where $\mc V$ is the vertex set of the graph. 
\end{theorem}

$F_k(v)$ can be given by a recurrence relation in~\cite{karp1978characterization}, where $\mc E$ denote  the edge set and $w(u,v)$ denotes the weight on $(u,v)$: 
\begin{equation*}
F_k(v) = \max_{(u,v)\in \mc E}[F_{k-1}(u) +w(u,v) ], \qquad k=1,2,\ldots,n
\end{equation*}
with the initial conditions $F_0(s)=0$ and $F_0(v)=-\infty$, $v\neq s$. 
The Karp's algorithm computes $F_k(v)$ recurrently for $k=0,1,\ldots,n $ and $v\in\mc V$. 
For more related discussions, we refer the readers to~\cite{karp1978characterization, dasdan1998faster}. 
Finally, we present our Algorithm~\ref{alg::sche_NC_2}.

\begin{algorithm}[tb]
\caption{Algorithm for Networks with Delays}
\label{alg::sche_NC_2}
\textbf{Input}: a network $\mc N=(\mc V, \mc L,\mc I, D)$

\textbf{Output}: maximum multiflow $v$	
\begin{algorithmic}[1]
\State Start with any rate vector $R_1\in \mc R$, $\mathcal{R}_1 \leftarrow \{R_1\}$
\State Construct $(\mc M_T, \mc E_T)$ from $\mc N$
\For{$i=1,2,\ldots$}
\State Run $v_i \leftarrow$LP-MMF$(\mc R_i)$ (or $v_i \leftarrow$LP-MCMF$(\mc R_i)$) and obtain the dual vector $\mu_i$
\State Construct $(\mc M_T, \mc E_T, w_{\mu_i})$ by $(\mc M_T, \mc E_T)$ and $\mu_i$
\State Find maximum-mean-cycle in $(\mc M_T, \mc E_T, w_{\mu_i})$ and obtain $R_{i+1} \leftarrow \argmax\limits_{R\in\mc R} \langle\mu_i, R\rangle $
\State $\mc R_{i+1} \leftarrow \conv \left(\mc R_i\cup \{R_{i+1}\}\right) $
\If{$\langle 
\mu_{i}, R_{i+1}
\rangle = \max_{R\in \mc R_{i}}
\langle 
\mu_{i}, R
\rangle $}\\
$\qquad $
\Return $v_i$
\EndIf
\EndFor
\end{algorithmic}
\end{algorithm}

\begin{theorem}[Optimality]
For a network $\mc N=(\mc V, \mc L,\mc I, D)$, Algorithm~\ref{alg::sche_NC_2} will terminate and output the maximum multiflow (or the maximum concurrent multiflow). 
\end{theorem}

The proof is similar to the proof of Theorem~\ref{thm::optimality_alg1}, though we use Lemma~\ref{lemma::weighted_SG} to justify that in iteration $i'$ when it terminates, the maximum-mean-cycle in $(\mc M_T, \mc E_T, \mu_{i'})$ gives 
\begin{equation*}
    R_{i'+1}=\argmax\limits_{R\in\mc R} \langle\mu_{i'}, R\rangle,
\end{equation*} and then we substitute it into  \begin{equation*}
    \langle \mu_{i'}, R_{i'+1} \rangle = \max_{R\in \mc R_{i'}} \langle \mu_{i'}, R \rangle,
\end{equation*}
which gives us the same result with~\eqref{eq:max_same}.
The remaining steps are similar with the proof of Theorem~\ref{thm::optimality_alg1} and hence are omitted.

\begin{remark}
Karp’s algorithm can find one maximum-mean-cycle in a time complexity polynomial in $|\mc E_T|$, the number of edges in the weighted scheduling graph. Nevertheless, $|\mc E_T|$ is exponential in $|\mathcal{L}|$, the number of links in the communication network. 
Therefore, the overall running time is at least exponential in $|\mathcal{L}|$, since we still desire the exact-optimal results. 
However, we note that finding the entire scheduling rate region generally has a time complexity exponential in $|\mc E_T|$, which should be doubly exponential in $|\mathcal{L}|$. 
As demonstrated in experiments in Section~\ref{sec::simu}, our algorithms can be significantly faster than the two-step method for both MMF and MCMF problems even in very simple networks, since we do not need the entire scheduling rate region. 
\end{remark}

We use $\mc N_{4,1}^{D=1}$ to illustrate our algorithm, which was also used in~\cite{ma2021rate, yang2023wireless_tit} in particular to illustrate their scheduling scheme. 
However, we focus on the weighted scheduling graph and the maximum-mean-cycle scheme, and our Algorithm~\ref{alg::sche_NC_2} based on them, which are not discussed in~\cite{ma2021rate, yang2023wireless_tit}.

\begin{example}
\label{eg::4_line}

We use $\mc N_{4,1}^{D=1}$ as an example (see Figure~\ref{fig:weighted_N41}). 
To first construct the scheduling graph, we find $T= \max\limits_{l\in \mc L} \max\limits_{l'\in \mc I(l)}|D(l,l')|=1 $ and denote the scheduling graph as $(\mc M_1, \mc E_1)$. 
The vertex set includes matrices of size $4\times 1$, each of which is a column of some collision-free schedules.

We have a vertex set $\mc M_1$ $ \{v_0,v_1,\ldots,v_8\}$, where
\begin{align}
&v_0 =\begin{bmatrix}
0\\  0\\  0\\  0
\end{bmatrix}, v_1 =\begin{bmatrix}
1\\  0\\  0\\  0
\end{bmatrix}, v_2 =\begin{bmatrix}
0\\  1\\  0\\  0
\end{bmatrix},  v_3 =\begin{bmatrix}
0\\  0\\  1\\  0 
\end{bmatrix},
v_4 =\begin{bmatrix}
0\\  0\\  0\\  1 
\end{bmatrix},\nonumber\\
& v_5 =\begin{bmatrix}
1\\  0\\  0\\  1
\end{bmatrix},  v_6 =\begin{bmatrix}
1\\  1\\  0\\  0
\end{bmatrix}, v_7 =\begin{bmatrix}
0\\  1\\  1\\  0
\end{bmatrix},v_8 =\begin{bmatrix}
0\\  0\\  1\\  1 
\end{bmatrix}. \label{eq::N41_vex}
\end{align}
The edge set $\mc E_1$ can be described as an adjacency matrix~\cite{ma2021rate}: 
\begin{equation}
\label{eq::N41_edge}
\begin{blockarray}{cccccccccc}
&v_0&v_1&v_2&v_3&v_4&v_5&v_6&v_7&v_8\\
\begin{block}{c[ccccccccc]}
v_0&1&1&1&1&1&1&1&1&1\\
v_1&1&1&0&1&1&1&0&0&1\\
v_2&1&1&1&0&1&1&1&0&0\\
v_3&1&1&1&1&0&0&1&1&0\\
v_4&1&1&1&1&1&1&1&1&1\\
v_5&1&1&0&1&1&1&0&0&1\\
v_6&1&1&0&0&1&1&0&0&0\\
v_7&1&1&1&0&0&0&1&0&0\\
v_8&1&1&1&1&0&0&1&1&0\\
\end{block}
\end{blockarray}.
\end{equation}
By implementing the algorithms in~\cite{ma2021rate, yang2023wireless_tit} we can find the scheduling rate region $\mc R$, whose vertices are: 
\begin{equation}
\label{eq::N41_region}
\begin{bmatrix}
1/2\\
1/2\\
1/2\\
1/2
\end{bmatrix},
\begin{bmatrix}
1/2\\
1/2\\
1/2\\
0
\end{bmatrix},
\begin{bmatrix}
0\\
1/2\\
1/2\\
1/2
\end{bmatrix},
\begin{bmatrix}
1\\
0\\
0\\
0
\end{bmatrix},
\begin{bmatrix}
0\\
1\\
0\\
0
\end{bmatrix},
\begin{bmatrix}
0\\
0\\
1\\
0
\end{bmatrix},
\begin{bmatrix}
0\\
0\\
0\\
1
\end{bmatrix},
\begin{bmatrix}
1\\
0\\
0\\
1
\end{bmatrix},
\begin{bmatrix}
0\\
0\\
0\\
0
\end{bmatrix}.
\end{equation}

Then we show that we can find the maximum flow without the need of the entire scheduling rate region~\eqref{eq::N41_region}.

Suppose we start with a rate vector $R_1=\begin{bmatrix}
0&1&0&0
\end{bmatrix}^\intercal$ and in the first iteration we solve the following linear program
\begin{align*}
\max  \quad& v=R(l_1)=R(l_2)=R(l_3)=R(l_4)\\
\text{s.t. } & R(l_i) \leq 0, \qquad i=1,3,4,\\
         & R(l_2) \leq 1, \\
         &  R(l_i) \geq 0, \qquad i=1,2,3,4.
\end{align*}
We obtain the maximum flow $ v=0$ and the corresponding dual vector is $\mu_1 = \begin{bmatrix}
\epsilon&0&\epsilon&0
\end{bmatrix}^\intercal$ with some $\epsilon>0$. 
We then use $\mu_1$ to construct the corresponding weighted scheduling graph. 
Consider the vertex set~\eqref{eq::N41_vex}, 
% and the edge set~\eqref{eq::N41_edge}, 
for $i=0,1,\ldots,8$, 
the weights on the edges of the weighted scheduling graph are: 
\begin{align*}
    w(v_i, v_0) = 0, \,\,\, w(v_i, v_1) = \epsilon,\,\,\, w(v_i, v_2) = 0,\\
    w(v_i, v_3) = \epsilon,\,\,\, w(v_i, v_4) = 0, \,\,\, w(v_i, v_5) = \epsilon,\\
    w(v_i, v_6) = \epsilon ,\,\,\, w(v_i, v_7) = \epsilon, \,\,\, w(v_i, v_8) = \epsilon.
\end{align*}
We can verify that there are $9$ vertices and $56$ edges. 
By the maximum-mean-cycle algorithm, we solve $\argmax_{R\in \mc R} \langle \mu_1, R_1\rangle$ and find $R_2 = \begin{bmatrix}
\frac{1}{2}&\frac{1}{2}&\frac{1}{2}&\frac{1}{2}
\end{bmatrix}^\intercal$. 
We can check that the algorithm terminates in the next step.

Compared to the vertices of $\mc R$ shown in~\eqref{eq::N41_region}, our approach only needs to find two of them. 
The initial starting point ($R_1$) is easy to find by searching an arbitrary maximal independent set of $\mc N_{4,1}^{D=1}$ (shown in Figure~\ref{fig:weighted_N41}). 
Moreover, we note that the maximum flow value $1/2$ indeed achieves the upper bound of line networks with utilizing nonzero delays, as proved in~\cite{bai2017throughput}. 
\end{example}

\section{Performance Evaluation}
\label{sec::simu}

In this section, we numerically demonstrate the efficiency of our algorithms, across various network settings
\footnote{All the implementations are based on Python, and executed on a laptop computer with i7-8550u CPU with Python 3.7.}. 
We should note that although various optimization frameworks~\cite{traskov2012scheduling, jones2012optimal, cui2007distributed, wiese2016scheduling, niati2012throughput} exist (see Section~\ref{subsec::review_opt} for a review) that can either output approximate solutions for general networks or exact solutions for restricted networks, our primary focus is on provably returning exact solutions for general networks, for a better understanding of the fundamental limits of general networks. 
Therefore, we compare our algorithm with the conventional two-step method (i.e., first calculate the entire scheduling rate region and then solve the MMF or MCMF problem), which guarantees to find optimal solutions on general networks like ours. 

\begin{itemize}
    \item We evaluate Algorithm~\ref{alg::sche_NC_1} on random networks; 

    \item We evaluate Algorithm~\ref{alg::sche_NC_2} on different line networks that have also been used and evaluated in~\cite{ma2021rate, yang2023wireless_tit, bai2017throughput}. 
\end{itemize}

Note that Algorithm~\ref{alg::sche_NC_2} can also be evaluated on random networks, but we choose to perform experiments on different line networks because: 
1) they are the networks discussed in the literature~\cite{ma2021rate, yang2023wireless_tit, bai2017throughput, chitre2012throughput, fan22isit}; 
2) due to the complexity of the scheduling part~\cite{ma2021rate, yang2023wireless_tit}, the two-step method is unable to handle networks even of moderate size, hence to compare with the two-step method we are somewhat restricted to small networks; 
3) we show that our algorithms yield significant improvements even in simple network settings.

\subsection{Algorithm~\ref{alg::sche_NC_1} on Random Networks}
\label{subsec::alg1_RandNet}

We first evaluate Algorithm~\ref{alg::sche_NC_1} on networks of $N$ ($3\leq N\leq 80$) nodes with random topology. 
As discussed in Section~\ref{sec::no_delay_net}, we let the networks be acyclic. 
For a given number of nodes, we randomly generate connected networks. 
We assume each link has a unit capacity, and we still use the $1$-hop interference model. 
We study \emph{multicast} case (in comparison,  existing works may only work on unicast case~\cite{jain2005impact, wan2009multiflows, zhou2013maximum}). 
We randomly choose $1$ node as the source and $2$ different nodes as the sink nodes. 
We do not specify any network coding mechanisms in particular, any network coding scheme can be directly applied (e.g., similar to~\cite{zhou2013maximum} we can utilize the COPE system~\cite{katti2006xors}). 
The main objective is to show the advantages of only calculating a subset of the scheduling rate region, while the conventional two-step method needs to calculate the scheduling region by searching maximal independent sets in the link conflict graph~\cite{jain2005impact} before sovling the MMF problem.

We compare the required time to solve the MMF problem by the two approaches. 
In Figure~\ref{fig:random_graph}, the required time (in logarithmic scale) is plotted against $3\leq N\leq 80$, the sizes of the random networks. 
Note for each $N$, we randomly generate an $N$-node network for $5$ times, which may have different topology. 
We randomly choose the source and sink nodes, and measure the average time used. 
Note the comparison is between the time used by two methods to find the optimal MMF values, and therefore the reduction in search time by our algorithm does not sacrifice any scheduling performance.

Figure~\ref{fig:random_graph} shows that our algorithm consistently outperforms the two-step method, and it can be significantly faster when the network size becomes large. 
When the network size is large, the scheduling region (which is a polytope) not only becomes hard to compute, but also leads to two difficulties: 1) the rate vectors (the vertices of the polytope) are calculated in the form of vectors, but to iteratively update the scheduling region, we need to convert the set of vectors ($V$-representation) to a set of inequalities ($H$-representation), and the conversion suffers high complexity; 2) the termination condition of the algorithm needs to decide whether a vector achieves an optimum, which can also have high cost when the scheduling region is large. 
Terminating in a small number of steps is important (which is the main advantage of our algorithms), since the update of a subset of the scheduling region is a critical bottleneck and should be performed as little as possible.

\begin{figure}
	\centering
    \includegraphics[scale = 0.27]{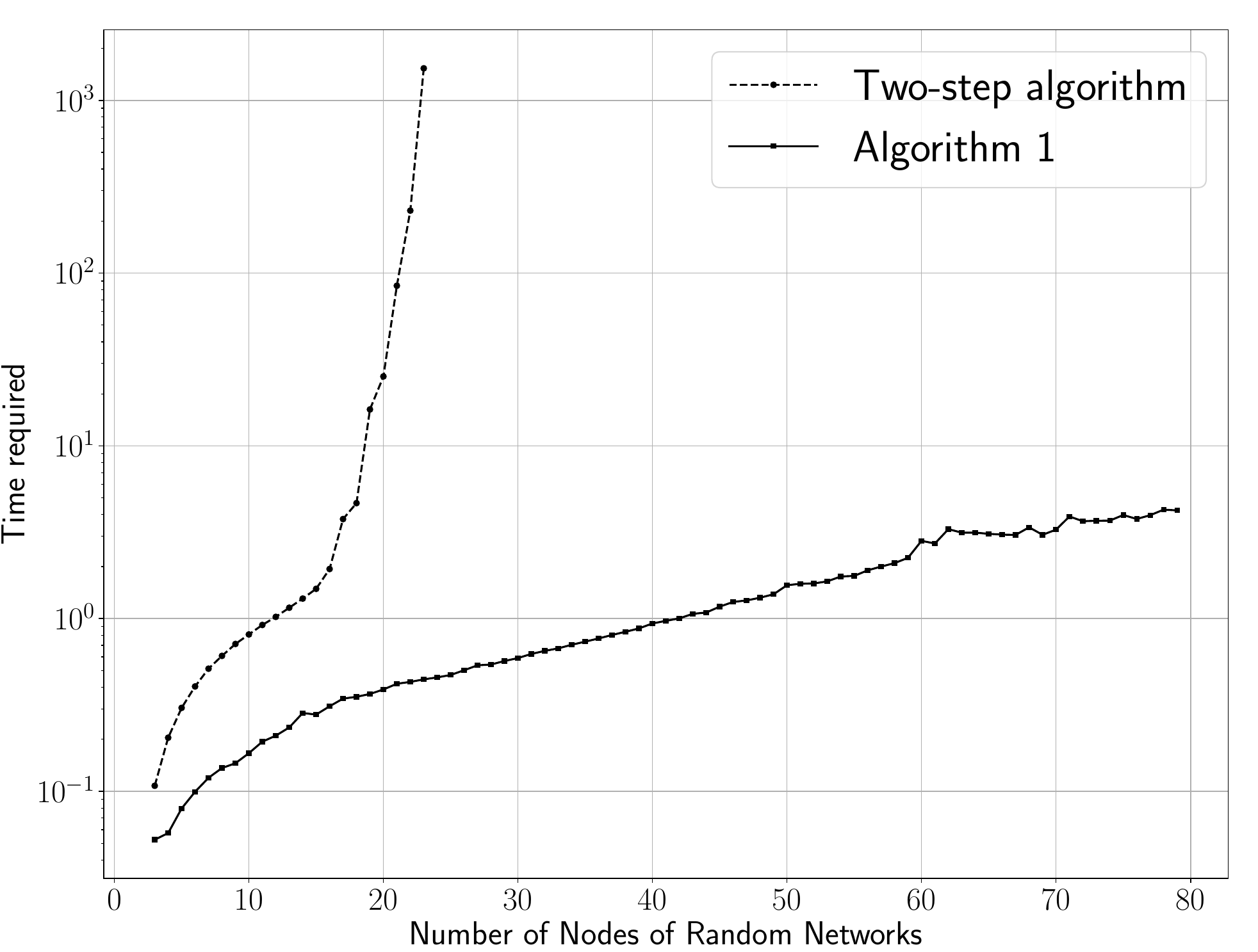}
	\caption{The $x$-axis is the number of nodes in random networks. The $y$-axis (in logarithmic scale) is the time (seconds) to calculate the MMF value. 
    } 
	\label{fig:random_graph}
\end{figure}

\subsection{Algorithm~\ref{alg::sche_NC_2} on Line Networks}

\begin{figure*}
	\centering
    \includegraphics[scale = 0.3]{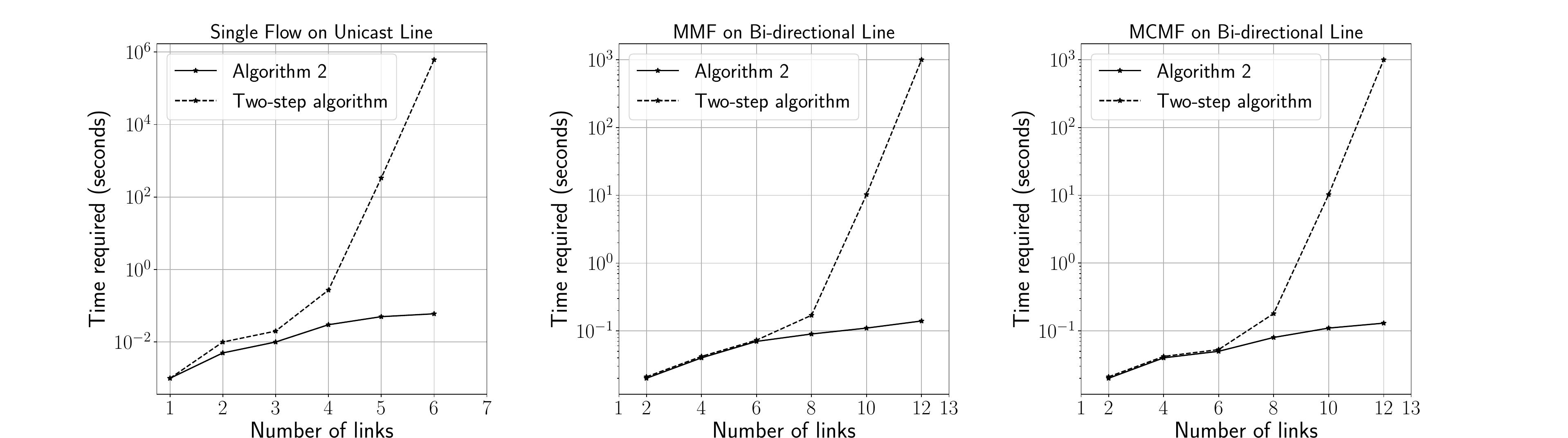}
	\caption{The $x$-axis is the number of links and the $y$-axis (in logarithmic scale) is the required time (seconds) to calculate the maximum (concurrent) multiflow value. 
    The left figure is for single-flow maximization in unicast line network $\mc N_{L,K}^{D=1}$; the middle figure is for the MMF problem in bi-directional line networks; the right figure is for the MCMF problem in bi-directional line networks. } 
	\label{fig:evaluation}
\end{figure*}

Next, we evaluate Algorithm~\ref{alg::sche_NC_2} on various line networks with non-negligible delays~\cite{ma2021rate, yang2023wireless_tit, bai2017throughput, chitre2012throughput} to show that not only do we outperform the conventional method based on~\cite{ma2021rate, yang2023wireless_tit}, but also that our algorithm can be significantly faster even in simple networks. 
Due to the hardness of the scheduling problem with non-negligible delays~\cite{ma2021rate, yang2023wireless_tit, bai2017throughput, chitre2012throughput}, the two-step solution can only handle networks of small sizes. 
We choose $D=1$ for simplification. 
We consider three tasks: single flow, MMF and MCMF maximizations:

\begin{enumerate}
    \item \textbf{Single flow maximization}:  
    Consider a unicast line network under the $1$-hop interference model $\mc N_{L,1}^{D=1}$. 
    The first node is the source, and the last node is the only sink.

    \item \textbf{MMF problem}: 
    Consider bi-directional line networks under the \emph{single collision domain} model~\cite{fan22isit, chitre2012throughput}, i.e., each link can collide with all the other links. 
    There are two information flows in the network, one from node $1$ to node $N$, and another from node $N$ to node $1$.

    \item \textbf{MCMF problem}: 
    We use the same bi-directional line networks with the same interference model. 
    For the two flows, we require that the desired traffic rate of one flow is half of the desired traffic rate of the other flow. 
\end{enumerate}

Similar to Section~\ref{subsec::alg1_RandNet}, we compare our algorithm with the two-step method. However, instead of enumerating the maximal independent sets in the conflict graph, we use the cycle-enumeration method~\cite{ma2021rate, yang2023wireless_tit} to calculate the scheduling rate region in the two-step method.

The performance evaluation is shown in Figure~\ref{fig:evaluation}. 
The running time (in logarithmic scale) is plotted against the number of links $L$. 
The figure demonstrates that the time required by our algorithms to achieve the optimal (concurrent) multiflow values is significantly lower. 
Since the comparison is between the time required to find the optimal MMF and MCMF values, the reduction in search time by our algorithm does not compromise scheduling performance. 
Additionally, it is evident that due to the NP-hardness of the scheduling problem, the required time for solving the multiflow problem by the two-step scheme in networks with non-negligible delays increases doubly-exponentially fast, which aligns with theoretical studies in~\cite{ma2021rate} and shows that the scheduling region is the key bottleneck, and hence we can have significant gains by only calculating a subset of the scheduling region.

For example, the algorithm in~\cite{ma2021rate} needs to search $7653$ rate vectors to characterize the scheduling region of $\mc N_{4,1}^{D=1}$, which has $9$ vertices as shown in~\eqref{eq::N41_region}, while we only need to find $\mathbf{2}$ of them for solving the MMF (or MCMF) problem. 
For $\mc N_{6,1}^{D=1}$ whose rate region has $57$ vertices, we only need $\mathbf{4}$ of them. 
The two-step method requires more than a week to calculate the scheduling rate region of $\mc N_{6,1}^{D=1}$, while Algorithm~\ref{alg::sche_NC_2} can solve the MMF problem in less than $\mathbf{1}$ \textbf{second}. 
In more complicated settings (e.g., random networks as discussed in Section~\ref{subsec::alg1_RandNet}), the scheduling problem becomes even harder, which might make our joint approach more preferable.

\section{Conclusion}
\label{sec::conclusion}

In large-scale wireless systems, although the maximum (concurrent) multiflow problem is extremely important for understanding network capacity, it is NP-hard even in simple network settings, making it computationally prohibitive to calculate. 
While there exist optimization frameworks that offer good performance and low complexity in practice, either providing approximate results for general networks or exact results for restricted networks, in this work we are more interested in a general approach that can always output exact-optimal solutions for general networks. 
In this paper, we provide algorithms that jointly solve the MMF and MCMF problems, as well as the scheduling problem, in a general multi-source multi-sink network with network coding allowed and propagation delays potentially utilized in scheduling. 
Our algorithms only require a subset of the scheduling rate region for the MMF and MCMF problems, making them much more efficient without sacrificing solution accuracy or scheduling performance. 
We theoretically prove that our algorithms output optimal solutions in a finite number of iterations and use simulation results to demonstrate the advantages of our approach.

\section{Future Works}
In Section~\ref{sec::delay_net}, we utilize the maximum-mean-cycle algorithms in solving the MMF (or MCMF) problem. 
We find that these algorithms can also be used to calculate the throughput (the objective in~\cite{bai2017throughput, chitre2012throughput}) or the entire scheduling rate region (the objective in~\cite{ma2021rate, yang2023wireless_tit}) as follows. 
The convex hull method~\cite{lassez1992quantifier} is an algorithm for finding the vertices of an unknown polytope $\mathcal{P} \subseteq \mathbb{R}^n$, given an oracle that can find $\mathrm{argmax}_{x \in \mathcal{P}}\langle x, a \rangle$ for $a \in \mathbb{R}^n$.
By using the maximum-mean-cycle algorithm as the oracle to the convex hull method, we can iteratively search for the vertices of the entire scheduling rate region. 
This method can be more efficient than the cycle-enumeration method in $(\mc M_T, \mc E_T)$~\cite{ma2021rate}, since we do not need to enumerate all the simple cycles in $(\mc M_T, \mc E_T)$. 
Detailed implementations of such algorithms for the scheduling problem is left for future study. 

\section{Acknowledgements}

The work of Shenghao Yang was partially supported by the National Key
R\&D Program of China under Grant 2022YFA1005000.
The work of Cheuk Ting Li was partially supported by
two grants from the Research Grants Council of the Hong
Kong Special Administrative Region, China [Project No.s:
CUHK 24205621 (ECS), CUHK 14209823 (GRF)]. 

\ifshortver
\newpage
\bibliographystyle{IEEEtran}
\bibliography{ref.bib}
\else
\bibliographystyle{IEEEtran}
\bibliography{ref.bib}
\fi

\end{document}